\RequirePackage{amsmath,amssymb}
\documentclass[a4paper]{llncs}
\usepackage[utf8]{inputenc}
\usepackage[T1]{fontenc}
\usepackage{tikz,microtype,etex,float,enumerate}

\usepackage[pdftex,pdfpagelabels,
  pdftitle={Towards Symbolic Factual Change in DEL},
  pdfsubject={Towards Symbolic Factual Change in DEL},
  pdfauthor={Malvin Gattinger}]{hyperref}
  \hypersetup{pdfborder = {0 0 0}}

\renewcommand{\phi}{\varphi}

\title{Towards Symbolic Factual Change in DEL}
\author{Malvin Gattinger}
\institute{Institute for Logic, Language \& Computation, University of Amsterdam}

\begin{document}

\setcounter{page}{14}
\pagestyle{plain}
\renewcommand{\thispagestyle}[1]{}

\maketitle

\begin{abstract}
We extend symbolic model checking for Dynamic Epistemic Logic (DEL) with factual
change. Our transformers provide a compact representation of action models with
pre- and postconditions, for both S5 and the general case.
The method can be implemented using binary decision diagrams and we expect it to
improve model checking performance.
As an example we give a symbolic representation of the Sally-Anne false belief
task.

\end{abstract}

\section{Introduction}

Symbolic representation is a solution to the state explosion problem in model
checking. The idea is to not store models explicitly in memory, but to find more
compact representations which still allow the evaluation of formulas.
In~\cite{BEGS15:SMCDEL} it was shown that S5 Kripke models for Dynamic Epistemic
Logic (DEL) can be encoded symbolically using \emph{knowledge structures}.
They are of the form $(V,\theta,O)$ where $V$ is a set of atomic propositions
called \emph{vocabulary}, $\theta$ is a boolean formula called \emph{state law}
and $O_i \subseteq V$ are observational variables for each agent.
Notably, this symbolic representation preserves the truth of all DEL formulas,
including higher-order knowledge (see Section~\ref{sec:basics} below).

The framework was generalized in~\cite{BEGS17:SMCDELbeyond} in two ways: From
equivalences to arbitrary relations and from announcements to action models. The
latter can be represented by \emph{knowledge transformers} of the form
$(V^+,\theta^+,O^+)$. Analogous to the product update on Kripke
models~\cite{BalMosSol98:tlopa}, applying a transformer to a structure yields a
new structure.
However, knowledge transformers only change what agents know and not what is the
case --- they do not provide a symbolic equivalent of postconditions for factual
change as studied in~\cite{BenEijKoo06:lcc}.

In this paper we combine the two generalizations and add the missing components
to treat factual change. The result are \emph{belief transformers with factual
change} which for simplicity we will just call \emph{transformers}.

Possible worlds in a Kripke model get their meaning but not their identity via
a valuation function. In particular we can assign the same atomic truths to
different possible worlds.
In contrast, all states of a knowledge structure satisfy different atomic
propositions and can thus be identified with their valuation. This is what
makes structures symbolic and efficient to implement, but it complicates the
idea of changing facts, as the following minimal example shows.

\begin{example}\label{ex:coin}
Consider a coin lying on a table with heads up: $p$ is true and this is common
knowledge. Suppose we then toss it randomly and hide the result from agent $a$
but reveal it to agent $b$. Figure~\ref{fig:KripkeChange} shows a Kripke model
of this update.

\begin{figure}\centering
  \begin{tikzpicture}[node distance=11mm]
    \node (m) [circle,draw] {$\begin{array}{c}w\\p\end{array}$};
    \draw (m) edge [loop left] node {a,b} (m);
  \node (times) [right of=m, node distance=15mm] {$\times$};
    \node (a) [right of=times, node distance=19mm] {};
    \node (a1) [draw,above of=a] {$\begin{array}{c}a_1\\?\top\\p:=\bot\end{array}$};
    \node (a2) [draw,below of=a] {$\begin{array}{c}a_2\\?\top\\p:=\top\end{array}$};
    \draw (a1) edge [loop left] node {a,b} (a1);
    \draw (a2) edge [loop left] node {a,b} (a2);
    \draw (a1) edge [<->] node [left] {a} (a2);
  \node (eq) [right of=a, node distance=16mm] {=};
    \node (m') [right of=eq, node distance=16mm] {};
    \node (wa1) [circle,draw,above of=m'] {$\begin{array}{c}(w,a_1)\\\lnot p\end{array}$};
    \node (wa2) [circle,draw,below of=m'] {$\begin{array}{c}(w,a_2)\\p\end{array}$};
    \draw (wa1) edge [loop right] node {a,b} (wa1);
    \draw (wa2) edge [loop right] node {a,b} (wa2);
    \draw (wa1) edge [<->] node [left] {a} (wa2);
  \end{tikzpicture}
  \caption{Factual change on Kripke models.}\label{fig:KripkeChange}
\end{figure}
\end{example}

It is easy to find the following structures that are equivalent to the initial
and the resulting model, but how can we symbolically describe the update which
transforms one into the other?
\[ \begin{array}{rc}
         & (V = \{p\},\ \theta = p,\ O_a=\{p\},\ O_b=\{p\}) \\
  \times & ??? \\
       = & (V = \{p\},\ \theta = \top,\ O_a=\varnothing\ ,\ O_b=\{p\}) \\
\end{array} \]

The name of a resulting world $(w,a_1)$ makes clear that it ``comes from'' $w$.
But a state like $\varnothing$ does not reveal its history or any relation to
$\{p\}$.
For purely epistemic actions this is not a problem --- we only \emph{add}
propositions from $V^+$ to the state to distinguish different epistemic events.
But for factual change propositions from $V$ have to be modified and we need a
way to \emph{remove} them from states.

Our solution is to \emph{copy propositions}: We store the old value of $p$ in a
fresh variable $p^\circ$. Then we rewrite the state law and observations using
substitutions.

We proceed as follows.
Sections~\ref{sec:KripkeStuff}~and~\ref{sec:basics} summarize the relevant parts
of~\cite{BEGS17:SMCDELbeyond}, generalized to belief transformers.
We then add factual change in Section~\ref{sec:factchange} and show that
transformers are equivalent to action models in Section~\ref{sec:equi}.
The Sally-Anne task illustrates our framework in Section~\ref{sec:sallyanne} and
we finish with further questions in Section~\ref{sec:future}.

\begin{definition}[Languages and Notation]\label{def:notation}
We fix a finite set of agents $I$ denoted by $i$, $j$, etc.\ and use the letters
$V$ or $X$ for sets of atomic propositional variables denoted by $p$, $q$, etc.

For any set of propositions $X$ we write $\mathcal{L}_B(X)$ for the boolean
language given by the BNF
  $\phi ::= p  \mid  \lnot \phi  \mid  \phi \land \phi$ where $p \in X$.
Similarly, let $\mathcal{L}(X)$ be the epistemic language over $X$ given by
  $\phi ::= p  \mid  \lnot \phi  \mid  \phi \land \phi  \mid  \Box_i \phi$
  where $i \in I$.

Primes and circles denote fresh variables, for example $p'$ and $p^\circ$.
For sets of variables let $X' := \{ x' \mid x \in X \}$ and $X^\circ := \{ x^\circ \mid x \in X \}$.
We also extend this notation to formulas recursively, for example $(\Box_i p \land q)' = (\Box_i p' \land q')$.

We write $[p/\psi]\phi$ for the result of substituting $\psi$ for $p$ in $\phi$.
Given two sets of the same size $A$ and $B$ of atomic propositions, and implicitly
assuming an enumeration $A=\{a_1,\dots,a_k\}$ and $B=\{b_1,\dots,b_k\}$ we write
$\left[B/A\right]\phi$ for the result of substituting $a_i$ for $b_i$ in
$\phi$ in parallel for all $i$.

A boolean assignment is identified with its set of true propositional variables
and we write $\vDash$ for the standard satisfaction relation.
Boolean quantification is used as follows:
$\forall p \phi : = \left[p/\top\right]\phi \wedge \left[p/\bot\right]\phi$.
For any $A=\{p_1,\dots,p_n\}$, let $\forall A \phi := \forall p_1 \forall p_2 \dots \forall p_n \phi$.
To abbreviate that a specific subset of propositions is true, let
$A \sqsubseteq B := \bigwedge A \land \bigwedge \{ \lnot p \mid p \in B \setminus A \}$.
\end{definition}

Definition~\ref{def:notation} describes operations on boolean \emph{formulas}
which might not be efficient in practice. In any actual implementation of our
methods those should be replaced with operations on boolean \emph{functions}
represented as Binary Decision Diagrams (BDDs)~\cite{Bryant86:bdds}.
For example, boolean quantification should \emph{not} be implemented as an
abbreviation but can instead be done efficiently by eliminating quantified
variables from the BDD\@.

\section{Kripke Models and Action Models}\label{sec:KripkeStuff}

We quickly state the standard definitions for Kripke semantics of DEL\@. For a
general introduction see~\cite{DitHoekKooi2007:del} and for details on factual
change see~\cite{BenEijKoo06:lcc}.

\begin{definition}
A \emph{Kripke model} for $V$ is a tuple $\mathcal{M}=(W,\mathcal{R},\pi)$
where $W$ is the set of worlds, $\mathcal{R}_i \subseteq W \times W$ is a
relation for each $i$ and $\pi : W \to \mathcal{P}(V)$ is a valuation function.
A \emph{pointed} Kripke model is a pair $(\mathcal{M},w)$ where $w \in W$.

We interpret $\mathcal{L}(V)$ on pointed Kripke models as follows.

\begin{enumerate}
  \item $(\mathcal{M},w) \vDash p$ iff $p \in \pi(w)$
  \item $(\mathcal{M},w) \vDash \neg \phi$ iff not $(\mathcal{M},w) \vDash \phi$
  \item $(\mathcal{M},w) \vDash \phi \land \psi$ iff $(\mathcal{M},w) \vDash \phi$ and $(\mathcal{M},w)\vDash \psi$
  \item $(\mathcal{M},w)\vDash \Box_i \phi \text{ iff }
      \text{for all } v \in W:
        \text{If } w \mathcal{R}_i v \text{ then } (\mathcal{M},v) \vDash \phi$
\end{enumerate}
\end{definition}

The following definition describes action models and how they can be applied
to Kripke models. Our definition of postconditions differs from the standard
in~\cite{BenEijKoo06:lcc} because we only allow boolean formulas. This however
does not change the expressivity~\cite{DitmarschKooi2008:SemanticResults}.

\begin{definition}\label{def:product}
An \emph{action model} is a tuple $\mathcal{A} = (A,R,\mathsf{pre},\mathsf{post})$ where
  $A$ is a set of atomic events,
  $R_i \subseteq A \times A$ a relation for each $i$,
  $\mathsf{pre} : A \to \mathcal{L}(V)$ is a precondition function and
  $\mathsf{post} : A \times V \to \mathcal{L}_B(V)$ a postcondition function.

The \emph{product update} is defined by $\mathcal{M} \times \mathcal{A} := (W^\text{new}, \mathcal{R}_i^\text{new}, \pi^\text{new})$ where
\begin{itemize}
  \item $W^\text{new} := \{ (w,a) \in W \times A \mid \mathcal{M},w \vDash \mathsf{pre}(a) \}$
  \item $\mathcal{R}_i^\text{new} := \{ ((w,a),(v,b)) \mid \mathcal{R}_i w v \text{ and } R_i a b \}$
  \item $\pi^\text{new}((w,a)) := \{ p \in V \mid \mathcal{M},w \vDash \mathsf{post}_a(p) \}$
\end{itemize}
An \emph{action} is a pair $(\mathcal{A},a)$ where $a \in A$.
\end{definition}

\section{Belief Structures and Belief Transformers}\label{sec:basics}

We now present the definitions of belief structures and belief transformers.
The key idea is that instead of explicitly listing worlds we use a symbolic
representation: The set of worlds and the valuation function are replaced by
a vocabulary and a boolean formula called the state law. The set of states is
then implicitly given as all boolean assignments that satisfy the state law,
i.e.~a subset of the powerset of the vocabulary.
Moreover, also the epistemic relations can be encoded using boolean formulas
and the goal is to interpret the language on the resulting structures without
ever computing or listing the full set of states.
For more details and proofs, we refer the reader to~\cite{BEGS17:SMCDELbeyond}.

\begin{definition}
A \emph{belief structure} is a tuple
$\mathcal{F} = (V,\theta,\Omega)$ where
  $V$ is a finite set of atomic propositions called \emph{vocabulary},
  $\theta \in \mathcal{L}_B(V)$ is the \emph{state law} and
  $\Omega_i \in \mathcal{L}_B(V \cup V')$ are called \emph{observations}.
Any $s \subseteq V$ such that $s \vDash \theta$ is called a \emph{state} of $\mathcal{F}$.
A pair $(\mathcal{F},s)$ where $s$ is a state of $\mathcal{F}$ is called a \emph{scene}.

We interpret $\mathcal{L}(V)$ on scenes as follows.
\begin{enumerate}
  \item $(\mathcal{F},s)\vDash p$ iff $ s\vDash p$.
  \item $(\mathcal{F},s)\vDash \neg \phi$ iff not $(\mathcal{F},s)\vDash \phi$
  \item $(\mathcal{F},s)\vDash \phi \land \psi$ iff $(\mathcal{F},s)\vDash \phi$ and $(\mathcal{F},s)\vDash \psi$
  \item $(\mathcal{F},s)\vDash \Box_i \phi \text{ iff }
      \text{for all } t \subseteq V:
        \text{If } t \vDash \theta \text{ and } (s \cup t') \vDash \Omega_i \text{ then } (\mathcal{F},t) \vDash \phi$.
\end{enumerate}

We write $(\mathcal{F},s) \equiv_V (\mathcal{F}',s')$
iff these two scenes agree on all formulas of $\mathcal{L}(V)$.
\end{definition}

An interesting property of belief structures is that on a given structure all
epistemic formulas have boolean equivalents. The following translation reduces
model checking to boolean operations which is not possible on Kripke models.

\begin{definition}\label{def:boolEquiv}
For any belief structure $\mathcal{F} = (V, \theta, \Omega)$ and any formula
$\phi \in \mathcal{L}(V)$ we define its \emph{local boolean translation}
$\| \phi \|_\mathcal{F}$ as follows.

\begin{enumerate}
  \item For any primitive formula, let $\| p \|_\mathcal{F} := p$.
  \item For negation, let $\| \neg \psi \|_\mathcal{F} := \neg \| \psi \|_\mathcal{F}$.
  \item For conjunction, let $\| \psi_1 \land \psi_2 \|_\mathcal{F} := \| \psi_1\|_\mathcal{F} \land \| \psi_2 \|_\mathcal{F}$.
  \item For belief, let $\| \Box_i\psi \|_\mathcal{F} := \forall V' ( \theta' \rightarrow ( \Omega_i \rightarrow (\| \phi \|_\mathcal{F})' ) )$
\end{enumerate}
\end{definition}

\begin{theorem}[from~\cite{BEGS17:SMCDELbeyond}]
Definition~\ref{def:boolEquiv} preserves and reflects truth.
That is, for any formula $\varphi$ and any scene $(\mathcal{F},s)$ we have that
  $(\mathcal{F},s)\vDash \varphi $ iff $s\vDash \| \varphi \|_\mathcal{F}$.
\end{theorem}

The following definition was only hinted at in~\cite{BEGS17:SMCDELbeyond}.
Belief transformers are like knowledge transformers, but instead of observed
propositions $O_i^+$ we use boolean formulas $\Omega_i^+$ to encode arbitrary
relations on $\mathcal{P}(V^+)$.

\begin{definition}\label{def:BelTransf}
A \emph{belief transformer} for $V$ is a tuple
  $ \mathcal{X} = (V^+, \theta^+, \Omega^+)$
where
  $V^+$ is a set of atomic propositions such that $V \cap V^+ = \varnothing$,
  $\theta^+ \in \mathcal{L}(V \cup V^+)$ is a possibly epistemic formula and
  $\Omega_i^+ \in \mathcal{L}_B(V \cup V^+)$ is a boolean formula for each $i \in I$.
A \emph{belief event} is a belief transformer together with a subset $x \subseteq V^+$, written as $(\mathcal{X},x)$.

The \emph{belief transformation} of a belief structure $\mathcal{F}=(V,\theta,\Omega)$
with $\mathcal{X}$ is defined by
  $\mathcal{F} \times \mathcal{X} := (V\cup V^+,\theta \land ||\theta^+||_\mathcal{F}, {\{ \Omega_i \land \Omega_i^+ \}}_{i \in I})$.
Given a scene $(\mathcal{F},s)$ and a belief event $(\mathcal{X},x)$, let
$(\mathcal{F},s) \times (\mathcal{X},x) := (\mathcal{F} \times \mathcal{X},s \cup x)$.
\end{definition}

The resulting observations are boolean formulas over the new double vocabulary
  $(V \cup V') \cup (V^+ \cup {V^+}') \ = \ (V \cup V^+) \cup (V \cup V^+)'$,
describing a relation between the new states which are subsets of $V \cup V^+$.

\section{Belief Transformers with Factual Change}\label{sec:factchange}

We now define transformation with factual change, adding the components $V_-$ and $\theta_-$.
Note that the belief transformers without factual change as discussed in the
previous section are exactly those transformers where $V_- = \varnothing$.

\begin{definition}\label{def:transform}
A \emph{transformer} for $V$ is a tuple
  $\mathcal{X} = (V^+, \theta^+, V_-, \theta_-, \Omega^+)$
where
\begin{itemize}
  \item $V^+$ is a set of fresh atomic propositions such that $V \cap V^+ = \varnothing$,
  \item $\theta^+$ is a possibly epistemic formula from $\mathcal{L}(V \cup V^+)$,
  \item $V_- \subseteq V$ is the \emph{modified subset} of the original vocabulary,
  \item $\theta_- : V_- \to \mathcal{L}_B(V \cup V^+)$ maps modified propositions to boolean formulas,
  \item $\Omega_i^+ \in \mathcal{L}_B(V^+ \cup {V^+}')$ are boolean formulas for each $i \in I$.
\end{itemize}
To transform $\mathcal{F}=(V,\theta,\Omega_i)$ with $\mathcal{X}$, let
  $\mathcal{F} \times \mathcal{X} := (V^\text{new}, \theta^\text{new}, \Omega_i^\text{new})$
where
\begin{enumerate}
  \item $V^\text{new} \ := \ V \cup V^+ \cup V_-^\circ$
  \item $\theta^\text{new} \ := \ \left[V_-/V_-^\circ\right]\left(\theta \land \|\theta^+\|_\mathcal{F} \right) \ \land \ \bigwedge_{q \in V^-} \left(q \leftrightarrow \left[V_-/V_-^\circ\right](\theta_-(q))\right)$
  \item $\Omega_i^\text{new} \ := \ \left(\left[V_-/V_-^\circ\right]\left[(V_-)'/(V_-^\circ)'\right]\Omega_i\right) \land \Omega_i^+$
\end{enumerate}

An \emph{event} is a pair $(\mathcal{X},x)$ where $x \subseteq V^+$.
Given $(\mathcal{F},s)$ and $(\mathcal{X},x)$, let
  $(\mathcal{F},s) \times (\mathcal{X},x) := (\mathcal{F} \times \mathcal{X}, s^\text{new})$
where the new actual state is
$s^\text{new} \ := \  (s \setminus V_-) \cup {(s \cap V_-)}^\circ \cup x \cup \{ p \in V_- \mid s \cup x \vDash \theta_-(p) \}$.
\end{definition}

To explain this definition, let us consider the components one by one.

First, the new vocabulary contains $V_-^\circ = \{ p^\circ \mid p \in V_- \}$.
These are fresh copies of the modified subset. We use them to keep track of the
old valuation.

Second, the new state law: A state in the resulting structure needs to satisfy
the old state law and the event law encoding the preconditions. For modified
propositions the old values have to be used, hence we apply a substitution to
both laws in the left conjunct.
Modified propositions are then overwritten in the right conjunct, using
$\theta_-$ which encodes postconditions.
As in Definition~\ref{def:product}, postconditions are evaluated in the old
model, hence we also substitute here.

Third, for the new observations we replace modified variables by their copies.
Two substitutions are needed because $\Omega_i^\text{new}$ is in a double vocabulary.
Old observations induce new ones via the state law. For example, if $q$ was
flipped publicly, then $q \leftrightarrow \lnot q^\circ$ is part of the new
state law and observing whether $q$ is equivalent to observing whether
$\lnot q^\circ$, i.e.~having observed $q$ in the original structure.
In the simpler S5 setting we would use
  $O_i^\text{new} \ := \ \left(\left[V_-/V_-^\circ\right]O_i\right) \cup O_i^+$.

Finally, the new actual state $s^\text{new}$ is the union of, in this order:
  propositions in the old state that have not been modified $(s \setminus V_-)$,
  copies of the modified propositions that were in the old state ${(s \cap V_-)}^\circ$,
  those propositions labeling the actual event $x$ and
  the modified propositions whose precondition was true in the old state $\{ p \in V_- \mid s \cup x \vDash \theta_-(p) \}$.

\begin{example}\label{ex:coinTransform}
We can now model the coin flip from Example~\ref{ex:coin} as follows.
Because we use the more general belief (instead of knowledge) structures, the
initial structure now has boolean formulas $\Omega_i$ instead of observational
variables $O_i$:
\[ (V = \{p\},\ \theta = p,\ \Omega_a = p \leftrightarrow p',\ \Omega_b = p \leftrightarrow p') \]
The following transformer models the coin flip visible to $b$ but not to $a$:
\[ (
  V^+ = \{q\},
  \ \theta^+ = \top,
  \ V_- = \{p\},
  \ \theta_-(p):= q,
  \ \Omega_a^+ = \top,
  \ \Omega_b^+ = q \leftrightarrow q'
) \]
The result of applying the latter to the former is this:
\[ (
  V = \{ p,q,p^\circ \},
  \ \theta = p^\circ \land (p \leftrightarrow q),
  \ \Omega_a = p^\circ \leftrightarrow {p^\circ}',
  \ \Omega_b =(p^\circ \leftrightarrow {p^\circ}') \land (q \leftrightarrow q')
) \]
\end{example}

\begin{example}
A publicly observable change $p:=\phi$ for a propositional formula $\phi$ is modeled by:
\[ (
  V^+ = \varnothing,
  \ \theta^+ = \top,
  \ V_- = \{p\},
  \ \theta_-(p):=\phi,
  \ \Omega_i^+=\top
) \]
\end{example}

DEL does not have temporal operators and agents never know the past explicitly.
Hence the old valuation is often irrelevant and the product update on Kripke
models does this ``garbage collection'' better than our transformation.
But we can eliminate propositions outside the original $V$ using the following Lemma.
A more thorough analysis of minimizing knowledge structures will be future work.

\begin{lemma}\label{lem:mini}
Suppose $\mathcal{F}$ uses the vocabulary $V \cup \{p\}$ and $p \not\in V$ is
determined by the state law (i.e.~$\theta \to p$ or $\theta \to \lnot p$ is a
tautology).
Then we can remove $p$ from the state law and observational BDDs to get a new
structure $\mathcal{F}'$ using the vocabulary $V$ such that
  $(\mathcal{F},s) \equiv_V (\mathcal{F}',s \setminus \{p\})$.
\end{lemma}

\begin{example}
The result from Example~\ref{ex:coinTransform} is $\equiv_{\{p,q\}}$ equivalent to:
  \[ (
    V = \{p,q\},
    \ \theta = p \leftrightarrow q,
    \ \Omega_a = \top,
    \ \Omega_b = q \leftrightarrow q'
  ) \]
\end{example}

\section{Equivalence and Expressiveness}\label{sec:equi}

We now show that transformers describe exactly the same class of updates as
action models. The main ingredients for the proof are the following Lemma and
two Definitions of how to go from transformers to action models and back.

\begin{lemma}[from~\cite{BEGS17:SMCDELbeyond}]\label{lem:morphism-nons5}
Suppose we have
  a belief structure $\mathcal{F}=(V,\theta,\Omega)$,
  a finite Kripke model $\mathcal{M}=(W,\pi,\mathcal{R})$ for the vocabulary $X \subseteq V$
  and a function $g: W \rightarrow \mathcal{P}(V) $ such that
\begin{enumerate}
  \item[C1] For all $w_1, w_2 \in W$ and $i \in I$ we have that
    $g(w_1)(g(w_2)') \vDash \Omega_i$ iff $w_1 \mathcal{R}_i w_2$.
  \item[C2] For all $w \in W$ and $p \in X$, we have that
    $p\in g(w)$ iff $p \in \pi(w)$.
  \item[C3] For every $s\subseteq V$,
    $s$ is a state of $\mathcal{F}$ iff $s=g(w)$ for some $w\in W$.
\end{enumerate}
Then, for every $\mathcal{L}(X)$-formula $\phi$ we have
  $(\mathcal{F},g(w))\vDash \phi$ iff $(\mathcal{M},w)\vDash \phi$.
\end{lemma}

\begin{definition}[$\mathsf{Act}$]\label{def:Act}
Given an event $(\mathcal{X} = (V^+, \theta^+, V_-, \theta_-, \Omega^+), x)$,
define an action $(\mathsf{Act}(\mathcal{X}) := (A,\mathsf{pre},\mathsf{post},R), a := x)$ by
\begin{itemize}
  \item $A := \mathcal{P}(V^+)$
  \item $\mathsf{pre}(a) := \left[a/\top\right]\left[(V^+ \setminus a)/\bot\right]\theta^+$
  \item $\mathsf{post}_a(p) := \left \{ \begin{array}{ll} \left[a/\top\right]\left[(V^+ \setminus a)/\bot\right]\left(\theta_-(p)\right) & \text{if } p \in V_- \\ p & \text{otherwise}\end{array} \right.$
  \item $R_i := \{ (a,b) \mid a \cup (b') \vDash \Omega_i^+ \}$
\end{itemize}
\end{definition}

\begin{definition}[$\mathsf{Trf}$]\label{def:Trf}
Consider an action $(\mathcal{A} = (A,\mathsf{pre},\mathsf{post},R),a_0)$.
Let $n := \mathsf{ceil}(\log_2|A|)$ and $\ell : A \to \mathcal{P}(\{q_1,\dots,q_n\})$ be an injective labeling function using fresh atomic variables $q_k$.
Then let $(\mathsf{Trf}(\mathcal{A}) := (V^+, \theta^+, V_-, \theta_-, \Omega^+),x:=\ell(a_0))$ be the event defined by
\begin{itemize}
  \item $V^+ := \{ q_1,\dots,q_n \}$
  \item $\theta^+ := \bigvee_{a \in A} \left( \mathsf{pre}(a) \land \ell(a) \sqsubseteq V^+ \right)$
  \item $V_- := \{ p \in V \mid \exists a : \mathsf{post}_a(p) \neq p \}$
  \item $\theta_-(p) := \bigvee_{a \in A} \left( \ell(a) \sqsubseteq V^+ \land \mathsf{post}_a(p) \right)$
  \item $\Omega_i^+ := \bigvee_{(a,b) \in R_i} \left( \ell(a) \sqsubseteq V^+ \land (\ell(b) \sqsubseteq V^+)' \right)$
\end{itemize}
\end{definition}

Besides these translations for the dynamic parts, we also use the translations
$\mathcal{M}(\cdot)$ and $\mathcal{F}(\cdot)$ from structures to models and
vice versa, as given in Definitions 18 and 19 of~\cite{BEGS17:SMCDELbeyond}.
Now everything is in place to state and prove our main result.
The following generalizes Theorem 4 in~\cite{BEGS17:SMCDELbeyond}.

\begin{theorem}
(i) Definition~\ref{def:Act} is truthful: For any scene $(\mathcal{F},s)$, any event $(\mathcal{X},x)$ and any formula $\phi$ over the vocabulary of $\mathcal{F}$ we have:
\[
  (\mathcal{F},s) \times (\mathcal{X},x) \vDash \phi
  \iff
  (\mathcal{M}(\mathcal{F}),s) \times (\mathsf{Act}(\mathcal{X}),x) \vDash \phi
\]
(ii) Definition~\ref{def:Trf} is truthful: For any pointed Kripke model $(\mathcal{M},w)$, any action $(\mathcal{A},a)$ and any formula $\phi$ over the vocabulary of $\mathcal{M}$ we have:
\[
  (\mathcal{M} \times \mathcal{A}, (w,a)) \vDash \phi
  \iff
  (\mathcal{F}(\mathcal{M}),g_\mathcal{M}(w)) \times (\mathsf{Trf}(\mathcal{A}),\ell(a))  \vDash \phi
\]
where $g_\mathcal{M}$ is from $\mathcal{F}(\mathcal{M})$ in Definition~19
of~\cite{BEGS17:SMCDELbeyond} and $\mathsf{Trf}(\mathcal{A})$ and $\ell$ are
from Definition~\ref{def:Trf} above.
\end{theorem}
\begin{proof}
By Lemma~\ref{lem:morphism-nons5}. We first need appropriate functions $g$.

For part (i), $g$ needs to map
worlds of $\mathcal{M}(\mathcal{F})\times\mathsf{Act}(\mathcal{X})$,
  i.e.~pairs $(s,x) \in \mathcal{P}(V) \times \mathcal{P}(V^+)$
to states of $\mathcal{F} \times \mathcal{X}$,
  i.e.~subsets of $V \cup V^+ \cup V_-^\circ$.
Let
  $g(s,x):= (s \setminus V_-) \cup {(s \cap V_-)}^\circ \cup x \cup \{ p \in V_- \mid s \cup x \vDash \theta_-(p) \}$
which is exactly $s^\text{new}$ from Definition~\ref{def:transform} above.
We now prove C1 to C3 from Lemma~\ref{lem:morphism-nons5}.

For C1, take any two worlds $(s,x)$ and $(t,y)$.
We need to show $g(s,x) (g(t,y))' \vDash \Omega_i^\text{new}$ iff $\mathcal{R}_i^\text{new}(s,x)(t,y)$.
For this, note the following equivalences.
We have $g(s,x) (g(t,y))' \vDash \Omega_i^\text{new}$ iff
\[ \begin{array}{rl}
       & (s \setminus V_-) \cup {(s \cap V_-)}^\circ \cup x \cup \{ p \in V_- \mid s \cup x \vDash \theta_-(p) \} \\
  \cup & ( (t \setminus V_-) \cup {(t \cap V_-)}^\circ \cup y \cup \{ p \in V_- \mid t \cup y \vDash \theta_-(p) \})' \\
\vDash & \left[V_-/V_-^\circ\right]\left[(V_-)'/(V_-^\circ)'\right]\Omega_i \land \Omega_i^+ \\
\end{array} \]
Here $V_-$ and $V_-'$ do not occur in the formula, as old epistemic relations
do not depend on new values of modified propositions.
Hence we can drop the subsets of $V_-$ and $V_-'$ to obtain the equivalent condition
\[ (s \setminus V_-) \cup {(s \cap V_-)}^\circ \cup x
  \cup (t \setminus V_-)' \cup (t^\circ \cap V_-^\circ)' \cup y'
  \vDash  \left[V_-/V_-^\circ\right]\left[(V_-)'/(V_-^\circ)'\right]\Omega_i \land \Omega_i^+ \]
in which we can split both sides into separate vocabularies:
\[ (s \setminus V_-) \cup {(s \cap V_-)}^\circ
  \cup (t \setminus V_-)' \cup (t^\circ \cap V_-^\circ)'
  \vDash   \left[V_-/V_-^\circ\right]\left[(V_-)'/(V_-^\circ)'\right]\Omega_i
  \text{ and }
  x \cup y' \vDash \Omega_i^+ \]
Now undo the $\circ$-substitution on both sides in the first conjunct to see that it is equivalent to $s \cup t' \vDash \Omega_i$.
Hence the whole condition is equivalent to
  $\mathcal{R}_i s t \text{ and } R_i x y$
which is exactly
  $ \mathcal{R}_i^\text{new}(s,x)(t,y)$
by Definition of $\mathcal{M}(\cdot)$ and Definition~\ref{def:Act}.

To show C2, take any $(s,x)$ and any $p\in V$.
We have to show that $p \in g(s,x)$ iff $p \in \pi^\text{new}(s,x) = \{ p \in V \mid \mathcal{M},s \vDash \mathsf{post}_x(p) \}$.
There are two cases.
  First, if $p \notin V_-$, then $\mathsf{post}_x(p)=p$ by Definition~\ref{def:Act} and we directly have
    $p \in g(s,x)$ iff $p \in s$ iff $\mathcal{M},s \vDash p$ iff $p \in \pi^\text{new}(s,x)$.
  Second, if $p \in V_-$, then
    $p \in g(s,x)$ iff $s \cup x \vDash \theta_-(p)$ by definition of $g$
    and $\mathsf{post}_x(p)= \left[x/\top\right]\left[(V^+ \setminus x)/\bot\right] \theta_-(p)$ by Definition~\ref{def:Act}.
    Hence we have a chain of equivalences: $p \in g(s,x)$
    iff $s \cup x \vDash \theta_-(p)$
    iff $s \vDash \left[x/\top\right]\left[(V^+ \setminus x)/\bot\right] \theta_-(p)$
    iff $\mathcal{M},s \vDash \left[x/\top\right]\left[(V^+ \setminus x)/\bot\right] \theta_-(p)$
    iff $p \in \pi^\text{new}(s,x)$.

For C3, take any $s^\text{new} \subseteq V \cup V^+ \cup V_-^\circ$.
  We want to show that $s^\text{new} \vDash \theta^\text{new}$ iff there is an $(s,x)$ such that $g(s,x) = s^\text{new}$.

For left-to-right, suppose $s^\text{new} \vDash \theta^\text{new}$,
  i.e.~$s^\text{new}  \vDash  \left[V_-/V_-^\circ\right]\left(\theta \land \|\theta^+\|_\mathcal{F} \right) \ \land \ \bigwedge_{q \in V^-} \left(q \leftrightarrow \left[V_-/V_-^\circ\right](\theta_-(q))\right)$.
  Now first, let $s^\text{old} := (s^\text{new} \cap \setminus V_-) \cup \{ p \in V_- \mid p^\circ \in s^\text{new} \}$.
  We then have $s^\text{old} \vDash \theta$, i.e.~$s^\text{old}$ is a state of $\mathcal{F}$ and thus by the definition of $\mathcal{M}(\cdot)$ also a world of $\mathcal{M}(\mathcal{F})$.
  Second, let $x := s^\text{new} \cap V^+$ and note that $s \cup x \vDash \| \theta^+ \|_\mathcal{F}$.
  It can now  be checked that $g(s,x)=s^\text{new}$.

For right-to-left, suppose we have an $(s,x)$ such that $g(s,x)=s^\text{new}$.
  Then we want to show
  $(s \setminus V_-) \cup {(s \cap V_-)}^\circ \cup x \cup \{ p \in V_- \mid s \cup x \vDash \theta_-(p) \}  \vDash  \left[V_-/V_-^\circ\right]\left(\theta \land \|\theta^+\|_\mathcal{F} \right) \ \land \ \bigwedge_{q \in V^-} \left(q \leftrightarrow \left[V_-/V_-^\circ\right](\theta_-(q))\right)$
  which indeed follows from $s \vDash \theta$ and Definition~\ref{def:Act}.

\smallskip

For part (ii), $g$ should map worlds of $\mathcal{M}\times\mathcal{A}$ to states of $\mathcal{F}(\mathcal{M}) \times {\mathsf{Trf}(\mathcal{A})}$.
Again we use $s^\text{new}$, but $s$ and $x$ are given by propositional encodings $g_\mathcal{M}(w)$ and $\ell(a)$.
Let $g(w,a):= (g_\mathcal{M}(w) \setminus V_-) \cup {(g_\mathcal{M}(w) \cap V_-)}^\circ \cup \ell(a) \cup \{ p \in V_- \mid g_\mathcal{M}(w) \cup \ell(a) \vDash \theta_-(p) \}$.
We leave checking C1 to C3 as an exercise to the reader --- the proofs are very
similar to those in part (i).
\qed
\end{proof}

\section{Symbolic Sally-Anne}\label{sec:sallyanne}

The Sally-Anne false belief task is a famous example used to illustrate and test for a theory of mind.
The basic version goes as follows (adapted from~\cite{BaronCohen198537}):

\begin{quote}
Sally has a basket, Anne has a box.
Sally also has a marble and puts it in her basket.
Then Sally goes out for a walk.
Anne moves the marble from the basket into the Box.
Now Sally comes back and wants to get her marble.
Where will she look for it?
\end{quote}

To answer this, one needs to realize that Sally did not observe that the marble
was moved and will thus look for it in the basket. We now translate the first
DEL modeling of this story from~\cite{Bolander2014:SeeingIsBelieving} to our framework.
This choice is also motivated by a recent interest in the complexity of theory
of mind~\cite{IvdPol:HowDiffToM,PolRooSzy15:ParCompToM} where our symbolic
representation might provide a new perspective.
For simplicity we adopt the naive modeling given in~\cite{Bolander2014:SeeingIsBelieving},
leaving it as future work to also adopt the refinement with edge-conditions
and other improvements of the model.

We use the vocabulary $V = \{p,t\}$ where
  $p$ means that Sally is in the room and
  $t$ that the marble is in the basket.
In the initial scene Sally is in the room, the marble is not in the basket and
both of this is common knowledge:
\[ (\mathcal{F}_0,s_0) = (( V=\{p,t\},\ \theta=(p \land \lnot t),\ \Omega_\text{S}=\top,\ \Omega_\text{A}=\top),\{p\}) \]

The sequence of events is:

\begin{itemize}
  \item[$\mathcal{X}_1$:] Sally puts the marble in the basket:
    $((\varnothing, \top, \{t\}, \theta_-(t)=\top, \top, \top), \varnothing)$.
  \item[$\mathcal{X}_2$:] Sally leaves:
    $((\varnothing, \top, \{p\}, \theta_-(p)=\bot, \top, \top), \varnothing)$.
  \item[$\mathcal{X}_3$:] Anne puts the marble in the box, not observed by Sally:\newline
    $((\{q\}, \top, \{t\}, \theta_-(t)= (\lnot q \to t) \land (q \to \bot), \lnot q', q \leftrightarrow q'), \{q\})$.
  \item[$\mathcal{X}_4$:] Sally comes back:
    $((\varnothing, \top, \{p\}, \theta_-(p)=\top, \top, \top), \varnothing)$.
\end{itemize}

We calculate the result in Figure~\ref{fig:sallyanne}, using Lemma~\ref{lem:mini} to remove superfluous variables.
Note that all operations are boolean.
Finally, we can check that in the last scene Sally believes the marble is in the basket:
\[ \begin{array}{rl}
     & \{p,q\} \vDash \Box_\text{S} t \\
\iff & \{p,q\}  \vDash  \forall V' (\theta' \to (\Omega_\text{S} \to t')) \\
\iff & \{p,q\}  \vDash  \forall \{p',t',q'\} ((t' \leftrightarrow \lnot q') \land p'  \to (\lnot q' \to t')) \\
\iff & \{p,q\}  \vDash \top
\end{array} \]

\begin{figure}
\[ \begin{array}{lll}
       & (( \{p,t\}, (p \land \lnot t), \top, \top),p) & \mathcal{F}_0 \\
\times & ((\varnothing, \top, \{t\}, \theta_-(t)=\top, \top, \top), \varnothing) & \mathcal{X}_1 \\
     = & (( \{ p,t,t^\circ \}, (p \land \lnot t^\circ) \land t, \top, \top),\{p,t\}) \\
\\
\times & ((\varnothing, \top, \{p\}, \theta_-(p)=\bot, \top, \top), \varnothing) & \mathcal{X}_2 \\
= & (( \{ p,t,t^\circ,p^\circ \}, (p^\circ \land \lnot t^\circ) \land t \land \lnot p, \top, \top),\{ t,p^\circ \}) \\
\equiv_V & (( \{p,t\}, t \land \lnot p, \top, \top),\{t\}) \\
\\
\times & ((\{q\}, \top, \{t\}, \theta_-(t)= (\lnot q \to t) \land (q \to \bot), \lnot q', q \leftrightarrow q'), \{q\}) & \mathcal{X}_3 \\
= & (( \{ p,t,q,t^\circ \}, t^\circ \land \lnot p \land (t \leftrightarrow ( (\lnot q \to t^\circ) \land (q \to \bot) ) ), \lnot q', q \leftrightarrow q'),\{q\}) \\
= & (( \{ p,t,q,t^\circ \}, t^\circ \land \lnot p \land (t \leftrightarrow \lnot q), \lnot q', q \leftrightarrow q'),\{q\}) \\
\equiv_V & (( \{p,t,q\}, \lnot p \land (t \leftrightarrow \lnot q), \lnot q', q \leftrightarrow q'),\{q\}) \\
\\
\times & ((\varnothing, \top, \{p\}, \theta_-(p)=\top, \top, \top), \varnothing) & \mathcal{X}_4 \\
= & (( \{ p,t,q,p^\circ \}, \lnot p^\circ \land (t \leftrightarrow \lnot q) \land p, \lnot q', q \leftrightarrow q'),\{p,q\}) \\
\equiv_V & (( \{p,t,q\}, (t \leftrightarrow \lnot q) \land p, \lnot q', q \leftrightarrow q'),\{p,q\}) \\
\end{array} \]
\caption{Sally-Anne on belief structures and transformers.}\label{fig:sallyanne}
\end{figure}

\section{Related and Future Work}\label{sec:future}

We generalized knowledge transformers from~\cite{BEGS17:SMCDELbeyond} to
\emph{belief transformers with factual change}. The result is a new symbolic
representation of action models with postconditions that can be implemented
using binary decision diagrams~\cite{Bryant86:bdds}.

As mentioned above, restricting postconditions to boolean formulas does not
limit the expressivity. The authors of~\cite{DitmarschKooi2008:SemanticResults}
in fact prove the stronger result that postconditions can be restricted to
$\top$ and $\bot$. Hence one can also model postconditions as functions of the
type $A \to \mathcal{P}(V)$ as done in~\cite{Bolander2014:SeeingIsBelieving}. We leave it
as future work to tune the definition of transformers in a similar way.

An alternative ``succinct'' representation for Kripke models and action models
was recently developed in~\cite{CarSchwar2017:succinctDEL}.
Succinct models also describe sets of worlds with boolean formulas, but instead
of observational variables or boolean formulas over a double vocabulary they
use \emph{mental programs} to encode accessibility relations.
Notably, model checking DEL is still in PSPACE when models and actions are
represented succinctly.
No complexity is known for our structures and transformers so far, but we
expect it to be the same as for succinct models and actions.

Finally, the presented ideas are of course meant to be implemented. A natural
next step therefore is to extend SMCDEL~\cite{MG:SMCDEL}, the implementation
of~\cite{BEGS17:SMCDELbeyond}, with the presented transformers, working towards
a symbolic model checker covering the whole \emph{Logic of Communication and
Change} from~\cite{BenEijKoo06:lcc}.
This work has been started and experimental modules including the Sally-Anne
example are now available at \url{https://github.com/jrclogic/SMCDEL}.

Additionally, an implementation of~\cite{CarSchwar2017:succinctDEL} would be
interesting to compare the performance of both approaches.
Benchmark problems can be taken from both the DEL and the cognition literature,
see for example~\cite{PolRooSzy15:ParCompToM}.

\subsubsection*{Acknowledgements.}
Many thanks to Fernando~R.~Velázquez~Quesada, Jan van Eijck and the anonymous
reviewers for helpful comments and suggestions.

\bibliographystyle{splncs03}
\bibliography{smfc}

\end{document}